\documentclass[11pt]{llncs}
\usepackage{yhmath}
\usepackage{url}
\usepackage{amssymb,wasysym}
\usepackage{graphicx}
\usepackage{wrapfig}
\usepackage{subfig}
\usepackage{color}
\usepackage{multirow}
\usepackage{epstopdf}
\usepackage[margin=1in]{geometry}

\usepackage[usenames,dvipsnames]{xcolor,colortbl}

\newcommand{\etal}{\emph{et~al.}}

\newcommand{\eg}{\emph{e.g.}}

\graphicspath{{images/}}

\urldef{\mailsa}\path||
\urldef{\mailsb}\path||
\urldef{\mailsc}\path||
\newcommand{\keywords}[1]{\par\addvspace\baselineskip
\noindent\keywordname\enspace\ignorespaces#1}

\renewenvironment{proof}{\emph{Proof.}}{\hfill $\Box$ \medskip\\}

\mainmatter  

\title{Kinetic Data Structures for the Semi-Yao Graph\\ and All Nearest Neighbors in $\mathbb{R}^d$}

\author{\small
Zahed Rahmati$^1$
\and
Mohammad Ali Abam$^2$
\and
Valerie King$^1$
\and
Sue Whitesides$^1$
}

\institute{\small
$^1$Dept. of Computer Science, University of Victoria, Canada. \\{\tt \{rahmati,val,sue\}@uvic.ca}\\
$^2$Dept. of Computer Engineering, Sharif University of Technology, Iran. \\{\tt abam@sharif.edu}
}

\toctitle{}
\tocauthor{}

\titlerunning{Kinetic Data Structures for the Semi-Yao Graph and All Nearest Neighbors}
\authorrunning{Z. Rahmati, M. A. Abam, V. King, and S. Whitesides}

\begin{document}
\maketitle
\begin{abstract}
This paper presents a simple kinetic data structure for maintaining all the nearest neighbors of a set of $n$ moving points in $\mathbb{R}^d$, where the trajectory of each point is an algebraic function of at most constant degree $s$.  The approach is based on maintaining the edges of the Semi-Yao graph, a sparse graph whose edge set includes the pairs of nearest neighbors as a subset. 

\vspace{+5pt}
Our kinetic data structure (KDS) for maintaining all the nearest neighbors is deterministic. It processes $O(n^2\beta_{2s+2}^2(n)\log n)$ events with a total cost of $O(n^2\beta_{2s+2}(n)\log^{d+1} n)$. Here, $\beta_s(n)$ is an extremely slow-growing function. The best previous KDS for all the nearest neighbors in $ \mathbb{R}^d$ is by Agarwal, Kaplan, and Sharir (TALG 2008). It is a randomized result. Our structure and analysis are simpler than theirs. Also, we improve their result by a factor of $\log^d n$ in the number of events and by a $\log n$ factor in the total cost.

\vspace{+5pt}
This paper generalizes and improves the 2013 work of Rahmati, King and Whitesides (SoCG 2013) on maintaining the Semi-Yao graph in $\mathbb{R}^2$;  its new technique provides the first KDS for the Semi-Yao graph in $\mathbb{R}^d$. Our KDS is local in the worst case, meaning that only a constant number of events is associated with any one point at any time.

\vspace{+5pt}
For maintaining all the nearest neighbors, neither our KDS nor the KDS by Agarwal~\etal~is local, and furthermore, each event in our KDS and in their KDS is handled in polylogarithmic time in an amortized sense. 

\vspace{+5pt}
Finally, in this paper, we also give a KDS for maintenance of all the $(1+\epsilon)$-nearest neighbors which is local and each event can be handled in a polylogarithmic worst-case time.
\end{abstract}
\newpage
\setcounter{page}{1}
\section{Introduction}
Let $P$ be a set of $n$ points in $R^d$. Finding the nearest neighbor to a query point, which was called the \textit{post office problem} by Donald Knuth in 1973~\cite{Knuth:1973:ACP}, is fundamental in computational geometry. The \textit{all nearest neighbors} problem, a variant of the post office problem, is to find the nearest neighbor to each point $p\in P$. Given any $\epsilon>0$, the \textit{all $(1+\epsilon)$-nearest neighbors} problem is to find some $\hat{q}\in P$ for each point $p\in P$, such that the Euclidean distance $|p\hat{q}|$ between $p$ and $\hat{q}$ is within a factor of $(1+\epsilon)$ of the Euclidean distance between $p$ and its nearest neighbor. The \textit{nearest neighbor graph} is constructed by connecting each point $p\in P$ to its nearest neighbor. The \textit{closest pair} problem is to find the endpoints of the edge in the nearest neighbor graph with minimum length. The \textit{Semi-Yao graph} (or \textit{theta-graph}) is a too long-well-studied sparse proximity graph~\cite{Clarkson:1987:AAS:28395.28402,Keil:1988:ACE:61764.61787}. This graph can be constructed by partitioning the space around each point $p\in P$ into $c$ cones $C_l(p)$, $1\leq l \leq c$, with $p$ the shared apex of the cones, and then connecting the point $p$ to a point $q$ inside each of these cones, such that the point $q\in P\cap C_l(p)$ has the minimum length projection on the vector in the direction of the symmetry $C_l(p)$; see Figure~\ref{fig:SY6}(b). 
By treating the number $c$ of cones at each point as a parameter of the Semi-Yao graph, one obtains an important class of sparse graphs, \textit{t-spanners}, with different stretch factors $t$~\cite{BBCRV2013Theta4,Bonichon:2010:CTD:1939238.1939265,BMRV2013Theta5,BRV2013Theta4k345}. 

The maintenance of attributes (e.g., the closest pair) of sets of moving points has been studied extensively over the past 15 years~\cite{DBLP:conf/swat/RahmatiZ12,Agarwal:2003:IMP:846156.846166,Alexandron:2007:KDD:1219156.1219201,Basch:1997:DSM:314161.314435,Karavelas:2001:SKG:365411.365441,DBLP:conf/gd/RahmatiZ12,DBLP:conf/iwoca/RahmatiZ11,DBLP:conf/socg12/Natin}. A basic framework for this study  is that of a kinetic data structure (KDS), which is in fact a set of data structures and algorithms to track the attributes of moving points. The problem of maintaining all the nearest neighbors, the closest pair, and the Semi-Yao graph on moving points are called the \textit{kinetic all nearest neighbors} problem, the \textit{kinetic closest pair} problem, and the \textit{kinetic Semi-Yao graph}, respectively. 

The kinetic maintenance of attributes is generally considered in two models: the \textit{standard KDS model} and the \textit{black-box KDS model}~\cite{Gao:2006:DSA:1646483.1646577,deBerg:2011:KCH:1998196.1998233}. In the black-box model, the locations of the objects are received at regular time steps. In the standard model, each object has a flight plan known in advance, and during the motion the object can change its flight plan at some times which are not known in advance. This paper considers the all nearest neighbors problem and the Semi-Yao graph in $\mathbb{R}^d$ for the standard KDS model and improves previous results; see Table~\ref{table:RelatedWork} in Appendix A. In addition, it offers results on all $(1+\epsilon)$-nearest neighbors in $\mathbb{R}^d$.
\vspace{+5pt}
\\
\textbf{Standard KDS framework.}
Basch, Guibas, and Hershberger~\cite{Basch:1997:DSM:314161.314435} introduced a \textit{kinetic data structure framework} to maintain the attributes of moving points in the standard KDS model. In this framework we assume each point $p\in P$ is moving with a trajectory  in $R^d$, where the $d$ coordinates of the trajectory of  the point $p$, which determine the position of $p$ as a function of time, are $d$ algebraic functions of at most constant degree $s$.  

The correctness of an attribute over time is determined based on correctness of a set of \textit{certificates}. A certificate is a boolean function of time, and its \textit{failure time} is the next time after the current time at which the certificate will become invalid. When a certificate fails, we say that an \textit{event} occurs. Using a \textit{priority queue} of the failure times of the certificates, we can know the next time after the current time that an event occurs. When the failure time of the certificate with highest priority in the priority queue is equal to the current time we invoke the update mechanism to reorganize the data structures and replace the invalid certificates with new valid ones. 

To analyse the performance of a KDS there are four standard criteria. A KDS distinguishes between two types of events: \textit{external events} and \textit{internal events}. An event that changes the desired attribute itself is called an external event and those events that cause only some internal changes in the data structures are called internal events. If the ratio between the total number of events and the number of external events is $O(\text{polylog}(n))$, the KDS is \textit{efficient}. The efficiency of the KDS gives an intuition of how many more events the KDS processes. If response time of the update mechanism to an event is $O(\text{polylog}(n))$, the KDS is \textit{responsive}. The compactness of a KDS refers to space complexity: if the KDS uses $O(n.\text{polylog}(n))$ space, it is \textit{compact}. The KDS is \textit{local} if the number of events associated with any point in the KDS is $O(\text{polylog}(n))$. The locality of the KDS is an important criterion; if a KDS satisfies locality, the KDS can be updated quickly and efficiently when an object changes its flight plan (trajectory) at times. 
\vspace{+5pt}
\\
\textbf{Related work.}
Basch, Guibas, and Hershberger (SODA'97)~\cite{Basch:1997:DSM:314161.314435} provided a KDS for maintenance of the closest pair in $\mathbb{R}^2$. Their KDS uses linear space and processes $O(n^2\beta_{2s+2}(n)\log n)$ events, each in time $O(\log^2 n)$. Here, $\beta_s(n)={\lambda_s(n)\over n}$ is an extremely slow-growing function and $\lambda_s(n)$ is the maximum length of Davenport-Schinzel sequences of order $s$ on $n$ symbols.

A common way to maintain attributes of moving points in $\mathbb{R}^d$ is to use \textit{kinetic multidimensional range trees}~\cite{Basch:1997:PPM:262839.262998,Agarwal:2008:KDD:1435375.1435379,Abam:2007:KKL:1247069.1247133}.  Basch~\etal~\cite{Basch:1997:PPM:262839.262998} and Agarwal~\etal~\cite{Agarwal:2008:KDD:1435375.1435379} use dynamic balanced trees to implement a kinetic range tree. Using rebalancing operations, they handle events to maintain a range tree. In particular, in their approaches, when an event between two points $p$ and $q$ occurs, we must delete $p$ and $q$ and reinsert them into the range tree. The range tree can be maintained over time using a dynamic range tree. One of the approaches to update the range trees is to carry out local and global rebuilding after a few operations, which gives $O(\log^d n)$ amortized time per operation~\cite{Mehlhorn:1984:DSA:1923}. Another approach, which uses merge and split operations, gives worst-case time $O(\log^d n)$ per operation~\cite{Willard:1985:ARR:3828.3839}. To avoid rebalancing the range tree after each operation, Abam and de Berg~\cite{Abam:2011:KSX:1971362.1971367} introduced a variant of the range trees, a \textit{rank-based range tree} (RBRT), which gives worst-case time $O(\log^d n)$ per operation.

Basch, Guibas, and Zhang (SoCG'97)~\cite{Basch:1997:PPM:262839.262998} used multidimensional range trees to maintain the closest pair. For a fixed dimension $d$, their KDS uses $O(n\log^{d-1}n)$ space and processes $O(n^2\beta_{2s+2}(n)\log n)$ events, each in worst-case time $O(\log^d n)$. Their KDS is responsive, efficient, compact, and local.

Agarwal, Kaplan, and Sharir (TALG'08)~\cite{Agarwal:2008:KDD:1435375.1435379} gave KDS's for both maintenance of the closest pair and all the nearest neighbors in $\mathbb{R}^d$. Agarwal~\etal~claimed that their closest pair KDS simplifies the certificates used by Basch, Guibas, and Hershberger~\cite{Basch:1997:DSM:314161.314435}; perhaps Agarwal~\etal~were not aware of the paper by Basch, Guibas, and Zhang~\cite{Basch:1997:PPM:262839.262998}, which independently presented a KDS for maintenance of the closest pair with the same approach to~\cite{Basch:1997:PPM:262839.262998}. The closest pair KDS by Agarwal~\etal, which supports insertions and deletions of points, uses $O(n\log^{d-1} n)$ space and processes $O(n^2\beta_{2s+2}(n)\log n)$ events, each in amortized time $O(\log^d n)$; this KDS is efficient, responsive (in an amortized sense), local, and compact. 
Agarwal~\etal~gave the first efficient KDS to maintain all the nearest neighbors in $\mathbb{R}^d$. For the efficiency of their KDS, they implemented multidimensional range trees by using randomized search trees (treaps). Their randomized kinetic approach uses $O(n\log^d n)$ space and processes $O(n^2\beta_{2s+2}^2(n)\log^{d+1} n)$ events; the expected time to process all events is $O(n^2\beta_{2s+2}^2(n)\log^{d+2} n)$. On average, each point in their KDS participates in $O(\log^d n)$ certificates. Their all nearest neighbors KDS is efficient, responsive (in an amortized sense), compact, but in general is not local. 

Rahmati, King, and Whitesides (SoCG'13)~\cite{socg17-rahmati} gave the first KDS for maintenance of the Semi-Yao graph in $\mathbb{R}^2$. Their Semi-Yao graph KDS uses linear space and processes $O(n^2\beta_{2s+2}(n))$ events with total processing time $O(n^2\beta_{2s+2}(n)\log n)$. Using the kinetic Semi-Yao graph, they improved the previous KDS by Agarwal~\etal~to maintain all the nearest neighbors in $\mathbb{R}^2$. In particular, their \textit{deterministic} kinetic algorithm, which is also arguably simpler than the randomized kinetic algorithm by Agarwal~\etal, uses $O(n)$ space and processes $O(n^2\beta_{2s+2}^2(n)\log n)$ events with total processing time $O(n^2\beta_{2s+2}^2(n)\log^2 n)$. With the same complexity as their KDS for maintenance of all the nearest neighbors, they maintain the closest pair over time. On average, each point in their KDS's participates in a constant number of certificates. Their kinetic data structures for maintenance of the Semi-Yao graph, all the nearest neighbors, and the closest pair are efficient, responsive (in an amortized sense), compact, but in general are not local.  
\vspace{+5pt}
\\
\textbf{Our technique and improvements.}
We provide a simple, deterministic KDS for maintenance of both the Semi-Yao graph and all the nearest neighbors in $\mathbb{R}^d$.  Assuming $d$ is fixed, we maintain the Semi-Yao graph in $\mathbb{R}^d$ using a constant number of range trees. Our KDS generalizes the previous KDS for the Semi-Yao graph by Rahmati~\etal~\cite{socg17-rahmati} that only works in $\mathbb{R}^2$. Also, our kinetic approach yields improvements of the KDS for maintenance of the Semi-Yao graph by Rahmati~\etal~\cite{socg17-rahmati} (see Table~\ref{table:RelatedWork} in Appendix A): Our KDS is local, but their KDS is not. In particular, each point in our KDS participates in $O(1)$ certificates, but  in their KDS each point participates in $O(n)$ events. For any fixed dimension $d$, our KDS handles $O(n^2)$ events, but their KDS handles $O(n^2\beta_{2s+2}(n))$ events in $\mathbb{R}^2$. 

Our KDS for maintenance of all the nearest neighbors is based on the fact that the Semi-Yao graph is a supergraph of the nearest neighbor graph. For each point $p$ in the Semi-Yao graph we construct a tournament tree to maintain the edge with minimum length among the edges incident to the point $p$. Summing over elements of all the tournament trees in our KDS is linear in $n$, which leads to a total number of events $O(n^2\beta_{2s+2}^2(n)\log n)$, which is \textit{independent} of $d$. This improves the previous \textit{randomized} kinetic algorithm by Agarwal~\etal~\cite{Agarwal:2008:KDD:1435375.1435379}: The expected total size of the tournament trees in their KDS is $O(n^2\beta_{2s+2}(n)\log^dn)$; so their KDS processes $O(n^2\beta_{2s+2}^2(n)\log^{d+1} n)$ events, which depends on $d$. Also, the structure and analysis by Agarwal~\etal~are more complex than ours.
\vspace{+5pt}
\\
\textbf{Our results.}
For a set of $n$ moving points in $\mathbb{R}^d$, we present simple KDS's to maintain the Semi-Yao graph and all the nearest neighbors. We assume the trajectory of each point is an algebraic function of at most constant degree $s$. Our KDS for maintenance of the Semi-Yao graph uses $O(n\log^d n)$ space and processes $O(n^2)$ events with total processing time $O(n^2\beta_{2s+2}(n)\log^{d+1} n)$.  The KDS is compact, efficient, responsive (in an amortized sense), and it is local.

Our all the nearest neighbors KDS uses $O(n\log^d n)$ space and processes $O(n^2\beta_{2s+2}^2(n)\log n)$) events with total processing time $O(n^2\beta_{2s+2}(n)\log^{d+1} n)$.  It is compact, efficient, responsive (in an amortized sense), but it is not local in general. To satisfy the locality criterion and to get a worst-case processing time KDS, we show a KDS for maintenance of all the $(1+\epsilon)$-nearest neighbors. In particular, for each point $p$ we maintain some point $\hat{q}$ such that $|p\hat{q}|<(1+\epsilon).|pq|$, where $q$ is the nearest neighbor of $p$ and $|pq|$ is the Euclidean distance between $p$ and $q$. This KDS uses $O(n\log^{d} n)$ space, and handles $O(n^2\log^d n)$ events, each in worst-case time $O(\log^d n\log\log n)$; it is compact, efficient, responsive, and local.
\vspace{+5pt}
\\
\textbf{Paper organization.}
Section~\ref{sec:prelininary} describes the construction of the Semi-Yao graph and gives a solution to the all nearest neighbors problem in $\mathbb{R}^d$. In Section~\ref{sec:KineticSY}, we show how the Semi-Yao graph can be maintained kinetically. Using the kinetic Semi-Yao graph, we give a KDS for maintenance of all the nearest neighbors in Section~\ref{sec:KineticANN}. Section~\ref{sec:KineticEpsANN} shows how to maintain all the $(1+\epsilon)$-nearest neighbors.
\section{The Construction}\label{sec:prelininary}
In the following we describe the construction of the Semi-Yao graph and construction of all the nearest neighbors, which will aid in understanding how our kinetic approach works.

Let $\overrightarrow{v}$ be a unit vector in $\mathbb{R}^d$ with apex at the origin $o$, and let $\theta$ be a constant. We define the \textit{infinite right circular cone} with respect to $\overrightarrow{v}$ and $\theta$ to be the set of points $x\in \mathbb{R}^d$ such that the angle between $\overrightarrow{ox}$ and $\overrightarrow{v}$ is at most $\theta/2$. A \textit{polyhedral cone} inscribed in this infinite right circular cone is formed by intersection of $d$ distinct half-spaces such that all the half-spaces contain the origin $o$. The angle between any two rays inside the polyhedral cone emanating from $o$ is at most $\theta$. 
The $d$-dimensional space around the origin $o$ can be covered by a collection of disjoint polyhedral cones $C_1,...,C_c$, where $c=O(1/\theta^{d-1})$~\cite{Agarwal:2008:KDD:1435375.1435379,Abam:2011:KSX:1971362.1971367}. Denote by $x_l$ the vector in the direction of the unit vector $\overrightarrow{v}$ of $C_l$, $1\leq l\leq c$,  with the origin at $o$. Let $C_l(p)$ denote a translated copy of $C_l$ with apex at $p$; see Figure~\ref{fig:SY6}(a). From now on, we assume $d$ is arbitrary but fixed, so $c$ is constant.

Given a point set $P$ in $\mathbb{R}^d$, the Semi-Yao graph is constructed by connecting each point $p\in P$ to the point in $P\cap C_l(p)$, $1\leq l\leq c$, whose $x_l$-coordinate is minimum. Figure~\ref{fig:SY6}(b) depicts some edges incident to the point $p$ in the Semi-Yao graph in $\mathbb{R}^2$ where $\theta=\pi/3$; here $x_1=-x_4$, $x_2=-x_5$, and $x_3=-x_6$.

The following lemma is used in~\cite{Basch:1997:DSM:314161.314435,Agarwal:2008:KDD:1435375.1435379,socg17-rahmati} to maintain the closest pair and all the nearest neighbors for a set $P$ of moving points; see Figure~\ref{fig:SY6}(c).
\begin{lemma}{\tt (Lemma 8.1.~\cite{Agarwal:2008:KDD:1435375.1435379})}\label{the:keyLemma}
Let $p$ be the nearest point to $q$ and let $C_l(p)$ be a cone of opening angle $\theta\leq \pi/3$ that contains $q$. Then $q$ has the minimum $x_l$-coordinate among the points in $P\cap C_l(p)$.
\end{lemma}
\begin{figure}[t]
\centering
\includegraphics[scale=0.7]{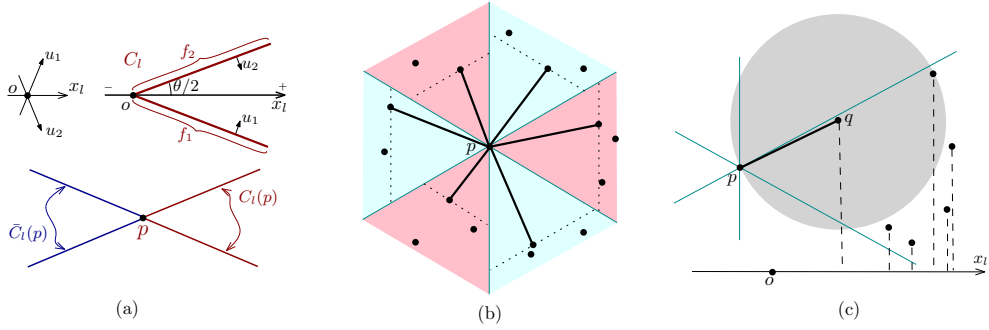}
\vspace{-5pt}
\caption{\small (a) The cone $C_l$ and its translated copy with apex at $p$. (b) The point $p$ is connected to the point in $P\cap C_l(p)$ that has minimum $x_l$-coordinate. The dotted lines are orthogonal to the cone axes. (c) The point $p$ is the nearest neighbor to $q$ and, so $q$ has the minimum $x_l$-coordinate among the points in $P\cap C_l(p)$.}
\vspace{-12pt}
\label{fig:SY6}
\end{figure}
For a set of points in the plane, Rahmati~\etal~\cite{socg17-rahmati} used Lemma~\ref{the:keyLemma} to show that the Semi-Yao graph is a super-graph of the nearest neighbor graph. It is easy to see the same result for a set of points in higher dimensions.
\begin{lemma}\label{the:SYcontainsNNG}
The Semi-Yao graph of a set of points in $\mathbb{R}^d$ is a super-graph of the nearest neighbor graph.
\end{lemma}
\begin{proof}
Let $(p,q)$ be an edge in the nearest neighbor graph such that $p$ is the nearest neighbor to $q$. The point $q$ is in some cone $C_l(p)$. Following the restriction $\theta\leq \pi/3$ of the cone $C_l(p)$ and from Lemma~\ref{the:keyLemma}, the point $q$ has minimum length projection on $x_l$ among the points in $P\cap C_l(p)$; this means that $(p,q)$ is an edge of the Semi-Yao graph.
\end{proof}

For a fixed dimension $d$, there is a constant number of cones $C_l$. Denote by $f_1,...,f_d$ the bounding half-spaces of the cone $C_l$ and let $u_i$ be the coordinate axis orthogonal to $f_i$, $1\leq i\leq d$; Figure~\ref{fig:SY6}(a) depicts $u_1$ and $u_2$ of the half-spaces $f_1$ and $f_2$ of the cone $C_l$. Corresponding to each cone $C_l$, we construct a ranked-based range tree (RBRT) $T_l$~\cite{Abam:2011:KSX:1971362.1971367}, which describes as follows, and for each point $p$ in ${\cal T}_l$, we find the point in $P\cap C_l(p)$ whose $x_l$-coordinate is minimum; this gives a construction for the Semi-Yao graph. 
\\
\textbf{Ranked-based range tree.} The RBRT $T_l$ is a variant of the range trees which has the following property. When two points exchange their order along an axis $u_i$, the RBRT can be updated without rebalancing the subtrees. The points at the level $i$ of the RBRT $T_l$ are sorted at the leaves in ascending order according to their $u_i$-coordinates. The skeleton of a RBRT is independent of the position of the points in $\mathbb{R}^d$ and it depends on the ranks of the points in each of the $u_i$-coordinates. The rank of a point in a tree at level $i$ of the RBRT is its position in the sorted list of all the points according to their $u_i$-coordinates. Any tree at any level of the RBRT is a balanced binary tree and no matter how many points are in the tree, it is a tree on $n$ ranks~\cite{Abam:2011:KSX:1971362.1971367}.

Let $v$ be an internal node at level $d$ of the RBRT $T_l$. Denote by $R(v)$ the set of points at the leaves of the subtree rooted at $v$.  The set $P\cap C_l(p)$ is the union of $O(\log^d n)$ sets $R(.)$; all of these $O(\log^d n)$ sets can be reported in time $O(\log^d n+k)$, where $k$ is the cardinality of $P\cap C_l(p)$. 

Corresponding to each node $v$ at level $d$ of $T_l$ we define another set $B(v)$. Denote by ${\cal P}_p$ the path from the parent of $p$ to the root of a tree at level $d$ of $T_l$. A point $p$ belongs to $B(v)$ if $v$ is the right child of some node $\bar{v}\in {\cal P}_p$; a point $p$ is in $B(v)$ if $R(v)$ is one of the $O(\log^d n)$ sets while reporting the points of $P\cap C_l(p)$.

Let $\bar{C}_l(p)=-C_l(p)$ be the reflection of $C_l(p)$ through $p$; $\bar{C}_l(p)$ is formed by following the lines through $p$ in the half-spaces of $C_l(p)$; see Figure~\ref{fig:SY6}(a). Similar to the way that we report the points of $P$ inside a query cone $C_l(q)$, we can also report the points of $P$ inside a query cone $\bar{C}_l(q)$. The set $P\cap \bar{C}_l(q)=\cup_v B(v)$, where the nodes $v$ are on the paths ${\cal P}_q$.


The set of all the pairs $(B(v),R(v))$, for all of the internal nodes $v$ at level $d$ of ${\cal T}_l$, is called a \textit{cone separated pair decomposition} (CSPD) for $P$ with respect to $C_l$; denote this set by $\Psi_{{C}_l}=\{(B_1,R_1),...,(B_m,R_m)\}$. The CSPD $\Psi_{{C}_l}$ for $P$ with respect to $C_l$ has the following properties~\cite{Abam:2011:KSX:1971362.1971367}:
\begin{itemize}
\item For the two points $p\in P$ and $q\in P$, where $q\in C_l(p)$, there exists a unique pair $(B_i,R_i)\in \Psi_{{C}_l}$ such that $p\in B_i$ and $q\in R_i$.
\item For the pair $(B_i,R_i)\in \Psi_{{C}_l}$, if $p\in B_i$ and $q\in R_i$, then $q\in C_l(p)$ and $p\in \bar{C}_l(q)$.
\end{itemize}
Let $r(v)$ be the point with minimum $x_l$-coordinate among the points in $R(v)$. Denote by $lc(v)$ and $rc(v)$ the left and the right child of the node $v$, respectively.  For each node $v$, the value of $r(v)$ is generated from the values of it children, $r(lc(v))$ and $r(rc(v))$; one which stores the point with minimum $x_l$-coordinate. So, for all internal nodes $v$ at level $d$ of the RBRT $T_l$, we can find all of the $r(v)$ in $O(n\log^{d-1} n)$ time. Since for each point $p\in P$ the point with minimum $x_l$-coordinate in $P\cap C_l(p)$ is chosen among $O(\log^d n)$ points $r(.)$, the following lemma results.
\begin{lemma}\label{the:SYConstructionTime}
The Semi-Yao graph of a set of $n$ points in $\mathbb{R}^d$ can be constructed in time $O(n\log^d n)$.
\end{lemma}
Vaidya~\cite{Vaidya:1989:ONL:70530.70532} gave an $O(n\log n)$ time algorithm to solve the all nearest neighbors problem. Given the Semi-Yao graph in $\mathbb{R}^d$, by examining the edges incident to any point in the Semi-Yao graph, we can find the nearest neighbor to the point. Since the Semi-Yao graph has $O(n)$ edges, we can get the following.
\begin{lemma}\label{the:ANNConstructionTime}
Given the Semi-Yao graph, the all nearest neighbors problem in $\mathbb{R}^d$ can be solved in $O(n)$ time.
\end{lemma}
\section{Kinetic Semi-Yao graph}\label{sec:KineticSY}
Fix a cone $C_l$ and the corresponding coordinates $u_1,...,u_d$, and $x_l$. When the points are moving, the Semi-Yao graph remains unchanged as long as the order of the points in each of the coordinates $u_1,...,u_d$, and $x_l$ remains unchanged. To maintain the Semi-Yao graph over time, we distinguish between two types of events:
\begin{itemize}
\item \textit{$u$-swap event:} Such an event occurs if two points exchange their order in the $u_i$-coordinate. This event can change the structure of the range tree.
\item \textit{$x$-swap event:} This event occurs if two points exchange their order in the $x_l$-coordinate. The range tree structure remains unchanged when this event occurs.
\end{itemize}
To track the above changes, we maintain sorted lists $L(u_1),...,L(u_d)$, and $L(x_l)$ of the points in each of the coordinates $u_1,...,u_d$, and $x_l$, respectively. For each two consecutive points in each sorted list $L(u_i)$ we define a certificate that certifies the order of the two points in the $u_i$-coordinate. To track the closest time to the current time we put failure times of all of the certificates in a priority queue; the element with the highest priority in the queue gives the closest time.

Our Semi-Yao graph KDS is based on maintenance of the RBRT ${\cal T}_l$. Abam and de Berg describe how to maintain a RBRT. Their approach uses $O(n\log^d n)$ space, and a $u$-swap event can be handled in the worst-case time $O(\log^d n)$ without rebalancing the subtrees of the RBRT~\cite{Abam:2011:KSX:1971362.1971367}.

For the point $w\in P$, the set $P\cap C_l(w)=\bigcup_{j=1}^{j=k} R(v_j)$ where the nodes $v_j$ are the right child nodes of the nodes on the paths ${\cal P}_w$. Denote by $\ddot{w}_l$ the point in $P\cap C_l(w)$ with minimum $x_l$-coordinate. To maintain the Semi-Yao graph, for each point $w\in P$ we must track $\ddot{w}_l$ which in fact is the point in $\{r(v_1),...,r(v_k)\}$ whose $x_l$-coordinate is minimum. To apply changes to the $\ddot{w}_l$, for all $w\in P$, besides $r(v)$, we need to maintain more information at each internal node $v$ at level $d$ of the RBRT ${\cal T}_l$ that describes as follows.


Allocate an \textit{id} to each point in $P$. Let $B'(v)=\{(w,\ddot{w}_l)|~w\in B(v)\}$ and let $L(B'(v))$ be a sorted list of the pairs of $B'(v)$ according to the ids of the second components $\ddot{w}$ of the pairs $(w,\ddot{w}_l)$. This sorted list is used to answer the following query while processing $x$-swap events: Given a query point $p$, find all the points $w\in B(v)$ such that $\ddot{w}_l=p$. Since we have insertions/deletions into the sorted lists $L(B'(.))$ over time, we implement them using a dynamic binary search tree (\eg, a \textit{red-black tree}); each insertion/deletion operation is performed in worst-case time $O(\log n)$. Furthermore, we maintain a set of links between each point $w\in P$ and the pair $(w,\ddot{w}_l)$ in the sorted lists $L(B'(.))$ where $w\in B(.)$; denote this set by $Link(w)$. Since, the point $w$ is in at most $O(\log^d n)$ sets $B(.)$, the cardinality of the set $Link(w)$ is $O(\log^d n)$.

In the preprocessing step before the motion, for any internal node $v$ at level $d$ of the RBRT and for any point $w\in P$, we find $r(v)$ and $\ddot{w}_l$ and then, we construct $L(B'(v))$ and $Link(w)$.

Now, let the points move. The following shows how to maintain and reorganize $Link(.)$, $L(B'(.))$ and $r(.)$  when a $u$-swap event or an $x$-swap event occurs. Note that maintenance of the sets $Link(w)$, for all $w\in P$, gives a kinetic maintenance of the Semi-Yao graph.

\begin{wrapfigure}{r}{0.4\textwidth}
\vspace{-25pt}
  \begin{center}
    \includegraphics[width=0.35\textwidth]{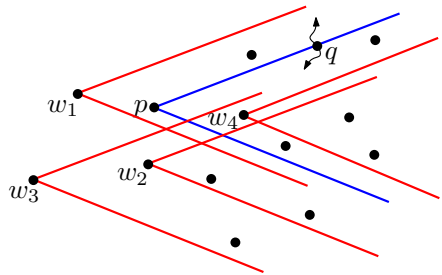}
  \end{center}
  \vspace{-10pt}
  \caption{\small A $u$-swap between $p$ and $q$ does not change the points in other cones $C_l(w_i)$.}
  \vspace{-15pt}
  \label{fig:Uswap}
\end{wrapfigure}
\textbf{$u$-swap event.}
When two points $p$ and $q$ exchange their order in the $u_i$-coordinate, we swap them in the sorted list $L(u_i)$ and update the invalid certificates with new valid ones; applying a constant number of changes to the priority queue takes $O(\log n)$ time. Then, we delete $p$ and $q$ and reinsert them into the RBRT with their new ranks~\cite{Abam:2011:KSX:1971362.1971367}. Next, we update the values $r(v)$ where the nodes $v$ are ancestors of $p$ and $q$. A change to some $r(v)$ can only change $r(v_{par})$ where $v_{par}$ is the parent of $v$. So, these updates can easily be done in $O(\log^d n)$ time.

Let $q\in C_l(p)$ (resp. $q\notin C_l(p)$) before the event. After occurring the event, $q$ moves outside (resp. inside) the cone $C_l(p)$; see Figure~\ref{fig:Uswap}. Note that this event does not change the points in cones $C_l(w)$ of other points $w\in P$. Therefore, the only change that can happen to the Semi-Yao graph is deleting an edge incident to $p$ inside the cone $C_l(p)$ and adding a new one. 

We perform the following steps when such event occurs. We first delete the pairs $(p,\ddot{p}_l)$ of the sorted lists $L(B'(.))$ where $p\in B(.)$; by using the links in $Link(p)$, this can be done in time $O(\log^{d+1})$. Then, we delete the members of $Link(p)$. Next, we find the point $\ddot{p}_l$ in $P\cap C_l(p)$ whose $x_l$-coordinate is minimum. Recall that $v_j$, $j=1,...,O(\log^d n)$, are the right child nodes of the nodes on the paths ${\cal P}_p$. Since we might get a new value for $\ddot{p}_l$ among all of the $r(v_j)$, we must add the new pair $(p,\ddot{p}_l)$, according to the id of the new value of $\ddot{p}_l$, into all the sorted lists $L(B'(v_j))$ where $p\in B(v_j)$. Finally, we construct $Link(p)$ of the new links between $p$ and the pair $(p,\ddot{p}_l)$ of the sorted lists $L(B'(v_j))$ which takes $O(\log^{d+1} n)$ time.

Since the number of swaps between the points in the sorted lists $L(u_i)$, $1\leq i \leq d$, is $O(n^2)$, the following results.
\begin{lemma}\label{the:Uswap}
For maintenance of the Semi-Yao graph, our KDS handles $O(n^2)$ $u$-swap events, each in the worst-case time $O(\log^{d+1} n)$.
\end{lemma}
\textbf{$x$-swap event.} Let $p$ and $q$ be two consecutive points with $p$ preceding $q$ in the sorted list $L(x_l)$: $x_l(p)<x_l(q)$ before the event. When $p$ and $q$ exchange their order, we swap them in $L(x_l)$ and update the invalid certificates with new valid ones, which takes $O(\log n)$ time. This event does not change the structure of the RBRT $T_l$; but it might change the second components of the pairs in some sorted lists $L(B'(.))$ and if so, we must apply the changes to the Semi-Yao graph. 

\begin{wrapfigure}{r}{0.4\textwidth}
\vspace{-20pt}
  \begin{center}
    \includegraphics[width=0.43\textwidth]{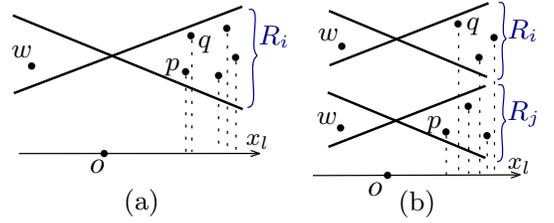}
  \end{center}
  \vspace{-10pt}
  \caption{\small Two cases when an $x$-swap between $p$ and $q$ occurs.}
  \vspace{-20pt}
  \label{fig:Xswap}
\end{wrapfigure}
The number of all changes to the Semi-Yao graph depends on how many points $w\in P$ have both of $p$ and $q$ in their cones $C_l(w)$. While reporting the points in $P\cap C_l(w)$, such points $w$ can have both of $p$ and $q$ in the same set $R_i$ (Figure~\ref{fig:Xswap}(a)) or in two different sets $R_i$ and $R_j$ (Figure~\ref{fig:Xswap}(b)). To find such points $w$, $(i)$ we seek internal nodes $v_{pq}$ at level $d$ of ${\cal T}_l$ where $\{p,q\}\subseteq R(v_{pq})$, which means that the nodes $v_{pq}$ are common ancestors of $p$ and $q$, $(ii)$ we seek for internal nodes $v_p$ and $v_q$ where $p\in R(v_p)$ and $q\in R(v_q)$. In the first case, it is obvious that we must find any point $w\in B(v_{pq})$ such that $p$ is the point with minimum $x_l$-coordinate in the cone $C_l(w)$, meaning that $\ddot{w}_l=p$. Then, we replace $p$ by $q$ after the event ($\ddot{w}_l=q$). This means that we replace the edge $wp$ of the Semi-Yao graph with $wq$. 

Note that in the second case there is no point $w'\in B(v_p)$ such that $\ddot{w'}_l=q$, because $x_l(p)<x_l(q)$. Also note that if there is a point $w''\in B(v_p)$ such that $\ddot{w''}_l=p$, we change the value of $\ddot{w''}_l$ to $q$ if $q\in C_l(w'')$. Since we can find $w''$ in $B(v_q)$, we do not need to be worry about checking whether such points $w''$ are in $B(v_p)$ or not. Therefore, for the second case, we only need to check whether there is a point $w\in B(v_q)$ such that $\ddot{w}_l=p$; if so, we change the value of $\ddot{w}_l$ to $q$ ($\ddot{w}_l=q$).

From the above discussion, the following three steps summarize the update mechanism of our KDS for maintenance of the Semi-Yao graph when an $x$-swap event occurs.
\begin{itemize}
\item[1.] Find all the internal nodes $v$ at level $d$ of $T_l$ such that $\{p,q\}\subseteq R(v)$ and $r(v)=p$; see Figure~\ref{fig:Xswap}(a). Also, find all the internal nodes $v$ where $r(v)=q$; see Figure~\ref{fig:Xswap}(b).
\item[2.] For each of the internal nodes $v$ (from Step 1), find all the pairs $(w,\ddot{w}_l)$ in the sorted list $L(B'(v))$ where $\ddot{w}_l=p$.
\item[3.] For each $w$ (from Step 2), through the links in $Link(w)$, find all the corresponding sorted lists $L(B'(.))$, delete the pair $(w,\ddot{w}_l)$ from them, change the value of the second component $\ddot{w}_l$ to $q$, and add $(w,\ddot{w}_l)$ into the sorted lists according to the id of $q$.
\end{itemize}
The number of edges incident to a point $p$ in the Semi-Yao graph is $O(n)$. So, when an $x$-swap event between $p$ and some point $q$ occurs, it might apply $O(n)$ changes to the Semi-Yao graph. The following lemma shows that an $x$-swap event can be handled in polylogarithmic amortized time.
\begin{lemma}\label{the:Xswap}
For maintenance of the Semi-Yao graph, our KDS handles $O(n^2)$ $x$-swap events with total processing time $O(n^2\beta_{2s+2}(n)\log^{d+1} n)$.
\end{lemma}
\begin{proof} 
All of the internal nodes $v$ at Step 1 can be found in $O(\log^d n)$ time. 

For each internal node $v$ of Step 2, the update mechanism spends $O(\log n + k_v)$ time where $k_v$ is the number of all the pairs $(w,\ddot{w}_l)\in B'(v)$ such that $\ddot{w}_l=p$. For all of the internal nodes $v$, the second step takes $O(\log^{d+1} n + \sum_v k_v)$ time. Note that $\sum_v k_v$ is equal to the number of exact changes to the Semi-Yao graph. Since, the number of changes to the Semi-Yao graph of a set of $n$ moving points in a fixed dimension $d$ is $O(n^2\beta_{2s+2}(n))$~\cite{socg17-rahmati}, the total processing time of Step 2 for all of the $O(n^2)$ $x$-swap events is $O(n^2\log^{d+1} n + n^2\beta_{2s+2}(n))=O(n^2\log^{d+1} n)$. 

The processing time to apply changes to the KDS for each $w$ of Step 3, which in fact is a change to the Semi-Yao graph, is $O(\log^{d+1} n)$. So, the update mechanism spends $O(n^2\beta_{2s+2}(n)\log^{d+1} n)$ time to handle all of the $O(n^2)$ events.

Hence, the total processing time for all the $x$-swap events is $O(n^2\beta_{2s+2}(n)\log^{d+1} n)$.
\end{proof}
The following summarizes the complexity of our Semi-Yao graph KDS.
\begin{theorem}\label{the:KineticSYG}
Our KDS for maintenance of the Semi-Yao graph of a set of $n$ moving points in $\mathbb{R}^d$, where the trajectory of each point is an algebraic function of at most constant degree $s$, uses  $O(n\log^d n)$  space and handles $O(n^2)$ events with a total cost of $O(n^2\beta_{2s+2}(n)\log^{d+1} n)$. The KDS is compact, efficient, responsive (in an amortized sense), and local.
\end{theorem}
\begin{proof}
From Lemmas~\ref{the:Uswap} and ~\ref{the:Xswap}, all of the $O(n^2)$ events can be processed in time $O(n^2\beta_{2s+2}(n)\log^{d+1} n)$; this means that the KDS is responsive in an amortized sense. $\sum_v|B(v)|=O(n\log^d n)$ and the number of the certificates is $O(n)$; so, the KDS is compact. A particular point in a sorted list $L(u_i)$ participates in two certificates, one certificate is created with previous point and one with the next point. So, the number of events associated to a particular point is $O(1)$ which means the KDS is local. Since the number of the external events is $O(n^2\beta_{2s+2}(n))$ and the number of the events that the KDS processes is $O(n^2)$, the KDS is efficient.
\end{proof}
\section{Kinetic All Nearest Neighbors}\label{sec:KineticANN}
Given the kinetic Semi-Yao graph, a super-graph of the nearest neighbor graph over time, from Section~\ref{sec:KineticSY}, we can easily maintain the nearest neighbor to each point $p\in P$. Let $Inc(p)$ be the set of edges in the Semi-Yao graph incident to the point $p$.  Using a \textit{dynamic and kinetic tournament tree} (DKTT)~\cite{Agarwal:2008:KDD:1435375.1435379,Basch:1997:DSM:314161.314435}, we can maintain the nearest neighbor to $p$ over time. Denote by ${\cal TT}_p$ the DKTT corresponding to the point $p$. The elements of ${\cal TT}_p$ are the edges in $Inc(p)$. The root of the ${\cal TT}_p$ maintains the edge with minimum length between the edges in $Inc(p)$. Let $m_p$ be the number of all insertions/deletions into the set $Inc(p)$ over time. 
\begin{lemma}\label{the:KineticTT} {\tt (Theorem 3.1.~\cite{Agarwal:2008:KDD:1435375.1435379})}
The dynamic and kinetic tournament tree ${\cal TT}_p$ can be constructed in linear time. For a sequence of $m_p$ insertions and deletions into the ${\cal TT}_p$, whose maximum size ${\cal TT}_p$ at any time is $n$, the ${\cal TT}_p$ generates at most $O(m_p\beta_{2s+2}(n)\log n)$ tournament events, for a total cost of $O(m_p\beta_{2s+2}(n)\log^2 n)$. Each event can be handled in the worst-case time $O(\log^2 n)$.
\end{lemma}
From Lemma~\ref{the:KineticTT}, the number of all events for maintenance of all the nearest neighbors is equal to $O(\sum_{p\in P} m_p\beta_{2s+2}(n)\log n)=O(\beta_{2s+2}(n)\log n\sum_{p\in P} m_p)$. Inserting (resp. deleting) an edge $pq$ in the Semi-Yao graph makes one insertion (resp. deletion) in ${\cal TT}_p$ and one in ${\cal TT}_q$. So, $\sum_{p\in P} m_p$ is in order of the number of exact changes to the Semi-Yao graph. Since $\sum_p|Ins(p)|=O(n)$, and since the number of all changes (edge insertions and edge deletions) to the Semi-Yao graph is equal to $O(n^2\beta_{2s+2}(n))$~\cite{socg17-rahmati}, the following corollary results.
\begin{corollary}\label{the:KineticANN}
Given a KDS for  maintenance of the Semi-Yao graph, all the nearest neighbors can be maintained by using a kinetic algorithm that generates $O(n^2\beta_{2s+2}^2(n)\log n)$ tournament events, for a total cost of $O(n^2\beta_{2s+2}^2(n)\log^2 n)$. Each event can be handled in the worst-case time $O(\log^2 n)$. The number of events associated to a particular point is constant on average.
\end{corollary}
The following theorem, which results from Theorem~\ref{the:KineticSYG} and Corollary~\ref{the:KineticANN}, gives the complexity of the KDS for maintenance of all the nearest neighbors.
\begin{theorem}\label{the:KineticAllNN}
Our kinetic data structure for maintenance of all the nearest neighbors of a set of $n$ moving points in $\mathbb{R}^d$, where the trajectory of each point is an algebraic function of at most constant degree $s$, has the following properties. 
\begin{enumerate}
\item The KDS uses $O(n\log^d n)$ space.
\item It processes $O(n^2)$ $u$-swap events, each in the worst-case time $O(\log^{d+1} n)$.
\item It processes $O(n^2)$ $x$-swap events, for a total cost of $O(n^2\beta_{2s+2}(n)\log^{d+1} n)$.
\item The KDS processes $O(n^2\beta_{2s+2}^2(n)\log n)$ tournament events, and processing all the events takes $O(n^2\beta_{2s+2}^2(n)\log^2 n)$ time.
\item The KDS is compact, efficient, responsive in an amortized sense, and local on average, meaning that each point participates in a constant number of certificates on average.
\end{enumerate}
\end{theorem}
\section{Kinetic All $(1+\epsilon)$-Nearest Neighbors}\label{sec:KineticEpsANN}
Let $q$ be the nearest neighbor of $p$ and let $\hat{q}$ be some point such that $|p\hat{q}|<(1+\epsilon).|pq|$. We call $\hat{q}$ the \textit{$(1+\epsilon)$-nearest neighbor} of $p$. In this section,  we maintain some $(1+\epsilon)$-nearest neighbor for any point $p\in P$.

\begin{wrapfigure}{r}{0.4\textwidth}
\vspace{-10pt}
  \begin{center}
    \includegraphics[width=0.35\textwidth]{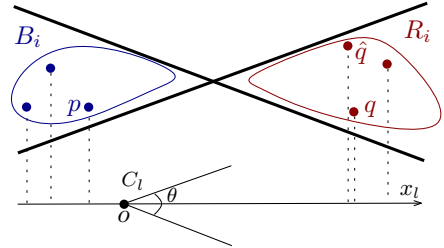}
  \end{center}
  \vspace{-15pt}
  \caption{\small The point $\hat{q}$ (resp. $p$) has minimum (resp. maximum) $x_l$-coordinate among the points in $R_i$ (resp. $B_i$).}
  \vspace{-10pt}
  \label{fig:RNNgraph}
\end{wrapfigure}
Consider a cone $C_l$ of opening angle $\theta$, which is bounded by $d$ half-spaces. Let $x_l$ be the  vector inside the cone $C_l$ that passes through the apex of $C_l$. Recall a CSPD $\Psi_{C_l}=\{(B_1,R_1),...,(B_m,R_m)\}$ for $P$ with respect to the cone $C_l$. Figure~\ref{fig:RNNgraph} depicts the cone $C_l$ and a pair $(B_i,R_i)\in \Psi_{C_l}$. Let $b_i$ (resp. $r_i$) be the point with the maximum (resp. minimum) $x_l$-coordinate among the points in $B_i$ (resp. $R_i$). Let $E_l=\{(b_i,r_i)|~i=1,...,m\}$. We call the graph $G(P,E_l)$ the \textit{relative nearest neighbor graph} (or RNN$_l$ graph for short) with respect to $C_l$. Call the graph $G(P,\cup_l E_l)$ the \textit{RNN graph}.  The RNN graph has the following interesting properties: $(i)$ It can be constructed in $O(n\log^d n)$ time by using a $d$-dimensional RBRT, $(ii)$ it has $O(n\log^{d-1} n)$ edges, and $(iii)$ the degree of each point is $O(\log^d n)$. Lemma~\ref{the:RNNGlemma} below shows another property of the RNN graph which leads us to find some $(1+\epsilon)$-nearest neighbor for any point $p\in P$.

\begin{lemma}\label{the:RNNGlemma}
Between all the edges incident to a point $p$ in the RNN graph, there exists an edge $(p,\hat{q})$ such that $\hat{q}$ is some $(1+\epsilon)$-nearest neighbor to $p$.
\end{lemma}
\begin{proof}
Let $q$ be the nearest neighbor to $p$ and let $q\in C_l(p)$. From the definition of a CSPD with respect to $C_l$, for $p$ and $q$ there exists a unique pair $(B_i,R_i)\in \Psi_{C_l}$ such that $p\in B_i$ and $q\in R_i$. From Lemma~\ref{the:keyLemma}, $p$ has the maximum $x_l$-coordinate among the points in $B_i$. 

Let $\hat{q}$ be the point with the minimum $x_l$-coordinate among the points in $R_i$. For any $\epsilon>0$, there exist an appropriate angle $\theta$ and a vector $x_l$ such that $|p\hat{q}|+(1+\epsilon).|q\hat{q}|\leq (1+\epsilon).|pq|$~\cite{Abam:2011:KSX:1971362.1971367}; this satisfies that $|p\hat{q}|\leq (1+\epsilon).|pq|$. Therefore, the edge $(p,\hat{q})$ which is an edge of the RNN graph gives some $(1+\epsilon)$-nearest neighbor.
\end{proof}

Consider the set $E_l$ of the edges of the RNN$_l$ graph. Let $N_l(p)=\{r_i|~(b_i,r_i)\in E_l~and~b_i=p\}$. Denote by $n_l(p)$ the point in $N_l(p)$ whose $x_l$-coordinate is minimum. Let $L(N_l(p))$ be a sorted list of the points in $N_l(p)$ in ascending order according to their $x_l$-coordinates; the first point in $L(N_l(p))$ gives $n_l(p)$. Similar to Section~\ref{sec:KineticSY} we handle two types of events, \textit{$u$-swap events} and \textit{$x$-swap events}. Note that we do not need to define a certificate for each two consecutive points in $L(N_l(.))$.  The following shows how to apply changes (\eg, insertion, deletion, and exchanging the order between two consecutive points) to the sorted lists $L(N_l(.))$ when an event occurs.

Each event can make $O(\log^d n)$ updates to the edges of $E_l$. Consider an updated pair $(b_i,r_i)$ that the value of $r_i$ (resp. $b_i$) changes from $p$ to $q$. For this update, we must delete $p$ (resp. $r_i$) form the sorted list $L(N_l(b_i))$ (resp. $L(N_l(p))$) and insert $q$ (resp. $r_i$) into $L(N_l(b_i))$ (resp. $L(N_l(q))$). If the event is an $x$-swap event, we must find all the $i$ where $r_i=q$ and check whether $n_l(b_i)=p$ or not; if so, $p$ and $q$ are in the same set $N_l(.)$ and we need to exchange their order in the corresponding sorted list $L(N_l(.))$. Since each update to a sorted list $L(N_l(.))$ can be done in $O(\log\log n)$, an event can be handled in worst-case time $O(\log^dn\log\log n)$.


From Lemma~\ref{the:RNNGlemma}, if the nearest neighbor of $p$ is in some cone $R_i$, then $r_i$ gives some $(1+\epsilon)$-nearest neighbor to $p$. Note that we do not know which cone $C_l(p)$, $1\leq l\leq c$, of $p$ contains the nearest neighbor of $p$, but it is obvious that the nearest point to $p$ among these $c$ points $n_1(p),...,n_c(p)$ gives some $(1+\epsilon)$-nearest neighbor of $p$. So, for all $l=1,...,c$, we track the distances of all the $n_l(p)$ to $p$ over time. A tournament tree (or a kinetic sorted list) of size $c$ with $O(1)$ certificates can be used to maintain the nearest point to $p$. Since each event makes $O(\log^d n)$ changes to the values of $n_l(.)$, and since the size of each tournament tree (or a kinetic sorted list) is constant, the number of all events to maintain all the $(1+\epsilon)$-nearest neighbors is $O(n^2\log^d n)$, and each point participates in $O(\log^d n)$ certificates.

From the above discussion the following theorem results.
\begin{theorem}\label{the:KinEpsANN}
Our KDS for maintenance of all the $(1+\epsilon)$-nearest neighbors of a set of $n$ moving points in $\mathbb{R}^d$, where the trajectory of each one is an algebraic function of constant degree $s$, uses $O(n\log^{d} n)$ space and handles $O(n^2\log^d n)$ events, each in the worst-case time $O(\log^d n\log\log n)$. The KDS is compact, efficient, responsive, and local.
\end{theorem}

Hence, for maintenance of all the $(1+\epsilon)$-nearest neighbors, as opposed to the all nearest neighbors KDS in Section~\ref{sec:KineticANN}, each event can be handled in a polylogarithmic worst-case time, and the KDS is local.
\bibliographystyle{splncs03}
\bibliography{References}
\newpage
\appendix
\section{The Comparison Table}
\begin{table}[h]
\small
\centering
\begin{tabular}{ | c | p{2cm} | c | c | c | p{1.5cm} |}
\cline{2-6}
\multicolumn{1}{ c| }{} & 
problem & space  &  total number of events &  total processing time&  locality\\ \cline{1-6}
  \multirow{2}{*}{Agarwal~\etal\cite{Agarwal:2008:KDD:1435375.1435379}} & \multirow{2}{*}{all NNs in $\mathbb{R}^d$} &  \multirow{2}{*}{$O(n\log^d n)$} & \multirow{2}{*}{$O(n^2\beta_{2s+2}^2(n)\log^{d+1} n)$} & \multirow{2}{*}{$O(n^2\beta_{2s+2}^2(n)\log^{d+2} n)$} & $O(\log^d n)$ on average\\  \hline\hline
\multirow{4}{*}{Rahmati~\etal\cite{socg17-rahmati}} &  \multirow{2}{*}{all NNs  in $\mathbb{R}^2$}& \multirow{2}{*}{$O(n)$} & \multirow{2}{*}{$O(n^2\beta^2_{2s+2}(n)\log n)$} & \multirow{2}{*}{$O(n^2\beta^2_{2s+2}(n)\log^2 n)$} & $O(1)$ on average\\  \cline{2-6} 
& Semi-Yao graph  in $\mathbb{R}^2$& \multirow{2}{*}{$O(n)$} & \multirow{2}{*}{$O(n^2\beta_{2s+2}(n))$} & \multirow{2}{*}{$O(n^2\beta_{2s+2}(n)\log n)$} & $O(1)$ on average\\  \hline\hline
\multirow{4}{*}{{{\color{Mahogany}This paper}}} & \multirow{2}{*}{all NNs  in $\mathbb{R}^d$} & \multirow{2}{*}{$O(n\log^d n)$} & \multirow{2}{*}{$O(n^2\beta_{2s+2}^2(n)\log n)$} & \multirow{2}{*}{$O(n^2\beta_{2s+2}(n)\log^{d+1} n)$} & $O(1)$ on average\\ \cline{2-6} 
& Semi-Yao graph in $\mathbb{R}^d$ & \multirow{2}{*}{$O(n\log^d n)$} & \multirow{2}{*}{$O(n^2)$} & \multirow{2}{*}{$O(n^2\beta_{2s+2}(n)\log^{d+1} n)$} & $O(1)$ in worst-case\\  \hline
\end{tabular}
\vspace{+10pt}
\caption{ The comparison between our KDS's and the previous KDS's, assuming $d$ is arbitrary but fixed.}
\label{table:RelatedWork}
\end{table}
\end{document}